%% file: robust_randomness_expansion.tex
\documentclass[11pt]{article}

\usepackage{amsmath, amsthm, amssymb, graphicx, enumerate, fullpage,amsfonts,mathrsfs,mathpazo,xspace}
\usepackage{color}
\usepackage{verbatim}
\usepackage{hyperref}

\usepackage{caption}
\usepackage{enumitem}
\usepackage{algpseudocode} 
\usepackage{algorithm}
\usepackage{algorithmicx}
\usepackage{float}

\makeatletter
\DeclareCaptionLabelFormat{protocol}{}
\makeatother

\theoremstyle{plain}
\newtheorem{thm}{Theorem}[section]
\newtheorem{theorem}[thm]{Theorem}
\newtheorem{definition}[thm]{Definition}
\newtheorem{proposition}[thm]{Proposition}

\newtheorem{lemma}[thm]{Lemma}

\newtheorem{claim}[thm]{Claim}
\newtheorem{fact}[thm]{Fact}

\theoremstyle{definition}

\newcommand{\A}{\mathcal{A}}
\newcommand{\B}{\mathcal{B}}
\newcommand{\X}{\mathcal{X}}
\newcommand{\Y}{\mathcal{Y}}

\newcommand{\SmoothMinEntropy}[2]{H^{#1}_\infty(#2)}
\newcommand{\MinEntropy}[1]{H_\infty(#1)}

\newcommand{\N}{\mathbb{N}}

\newcommand{\eps}{\varepsilon}

\newcommand{\win}{\textsc{win}}
\newcommand{\CHSH}{\textsc{chsh}}
\newcommand{\magic}{\textsc{ms}}

\newcommand{\valq}{\omega_q}

\newcommand{\valns}{\omega_{ns}}

\definecolor{mygrey}{gray}{0.50}

\begin{document}

\title{Robust Randomness Amplifiers: Upper and Lower Bounds}
\author{
\makebox[.2\columnwidth]{Matthew Coudron\thanks{Computer Science and Artificial Intelligence Laboratory, Massachusetts Institute of Technology. T.V. was supported by the National Science Foundation under Grant No. 0844626 and by and by the Ministry of Education, Singapore under the Tier 3 grant MOE2012-T3-1-009. H.Y. was supported by an NSF Graduate Fellowship Grant No. 1122374 and National Science Foundation Grant No. 1218547.  M.C. was supported by the National Science Foundation under Grant No. 0801525.}}
\and 
\makebox[.2\columnwidth]{Thomas Vidick\footnotemark[1]\,\,\,\thanks{Centre for Quantum Technologies, Singapore}}
\and 
\makebox[.2\columnwidth]{Henry Yuen\footnotemark[1]}
}

\maketitle

\begin{abstract}
A recent sequence of works, initially motivated by the study of the nonlocal properties of entanglement, demonstrate that a source of information-theoretically certified randomness can be constructed based only on two simple assumptions: the prior existence of a short random seed and the ability to ensure that two black-box devices do not communicate (i.e. are non-signaling). We call protocols achieving such certified amplification of a short random seed \emph{randomness amplifiers}.

We introduce a simple framework in which we initiate the systematic study of the possibilities and limitations of randomness amplifiers. Our main results include a new, improved analysis of a robust randomness amplifier with exponential expansion, as well as the first upper bounds on the maximum expansion achievable by a broad class of randomness amplifiers. In particular, we show that non-adaptive randomness amplifiers that are robust to noise cannot achieve more than doubly exponential expansion. Finally, we show that a wide class of protocols based on the use of the $\CHSH$ game can only lead to (singly) exponential expansion if adversarial devices are allowed the full power of non-signaling strategies. Our upper bound results apply to all known non-adaptive randomness amplifier constructions to date.
\end{abstract}

\section{Introduction}\label{sec:intro}
\input{intro}

\section{Preliminaries}\label{sec:prelim}
\input{prelim}

\section{Randomness amplifiers}\label{sec:model}
\input{model}

\section{Lower bounds}\label{sec:lower}
\input{lower_bounds}

\section{Upper bounds}\label{sec:upper}
\input{upper_bounds}

\appendix
\input{appendices}

\bibliographystyle{alphaabbrvprelim}
\bibliography{randomness}

\end{document}

%% file: intro.tex
Consider the following simple game, called the $\CHSH$ game: a referee sends each of a pair of isolated, cooperating but non-communicating players Alice and Bob a bit $x,y\in\{0,1\}$ respectively, chosen uniformly at random. Alice and Bob reply with bits $a,b\in\{0,1\}$, and they win the game iff $a\oplus b = x \wedge y$. If Alice and Bob employ classical strategies, the probability that they win the game is at most $75\%$. As a consequence, one readily sees that \emph{any} non-signaling strategy (i.e. a strategy in which each player's marginal output distribution is independent of the other player's input) that wins the $\CHSH$ game with probability strictly larger than $75\%$ \emph{must} generate randomness. Remarkably, there actually \emph{exists} such a strategy, allowing them to win with probability $\cos^2(\pi/8) \approx 85\%$. Furthermore, the strategy can be physically implemented using simple ``everyday'' quantum mechanical devices that utilize shared entanglement~\cite{Aspect81}. In his Ph.D. thesis, Colbeck~\cite{Colbeck09} was the first to explicitly observe that the $\CHSH$ game could be interpreted as a simple \emph{statistical test} for the presence of randomness: the test repeatedly ``plays'' the CHSH game with a given pair of black-box devices. Provided that non-signaling is enforced between the devices (via space-time separation or other means), the observation of a sufficiently high success probability can be used to \emph{certify} the generation of ``fresh'' randomness. In particular, the soundness of the test does \emph{not} require one to assume that quantum mechanics is correct. (Of course, as far as we know, the easiest way to actually \emph{pass} the test is to perform certain specific quantum mechanical measurements on two halves of an EPR pair!)

It is easy to see that without any assumptions, black-box randomness testing is impossible: if a (randomized) test $T$ accepts a random source $X$ with some probability $p$, by linearity of expectation there automatically exists a deterministic source $Y$ (i.e. a fixed string) that is accepted with probability at least $p$. Thus it is quite surprising that a very simple physical assumption -- that it is possible to enforce non-signaling between two devices -- allows for an information-theoretic method to test for randomness in the devices' outputs. As we shall see, the test provides guarantees on the \emph{min-entropy} of the outputs, which enables the tester to later apply a classical procedure such as a randomness extractor to generate bits that are nearly independent and uniformly distributed, making them useful in algorithmic or cryptographic applications (for a survey on randomness extractors we refer to~\cite{Sha02}).

Starting with work of Pironio et al.~\cite{Pironio}, a series of papers~\cite{Colbeck11,FGS11,vazirani2012certifiable,PM11} have demonstrated that not only can randomness be certified, but it can be \emph{expanded} as well. In~\cite{Pironio}, a protocol was given in which the testing requires $m$ bits of seed randomness, but the output of the devices is certified to have $\Omega(m^2)$ bits of min-entropy. Vazirani and Vidick~\cite{vazirani2012certifiable} show that there exists a protocol that can produce $2^{\Omega(m)}$ bits of certifiable randomness starting from $m$ bits of seed randomness. In their protocol, the referee uses the seed to generate pseudorandom inputs for the two devices; the devices play $2^{O(m\log^2m)}$ iterations of (a variant of) the $\CHSH$ game on those inputs. The referee then tests that the wins and losses of the devices obey a simple statistical condition. One can show that, whenever the devices are designed in a way that they pass the test with non-negligible probability, their output distribution (conditioned on passing) must have high min-entropy. The test, however, is not \emph{robust} in the sense that even a very slight deviation by the devices from the intended behavior will result in rejection. Robust protocols for exponential randomness expansion were devised in~\cite{FGS11,PM11} but they use \emph{two} pairs of devices, and furthermore rely on the strong assumption that there is no entanglement between the pairs. 

These prior works immediately raise a wealth of questions, for which there has been no systematic investigation so far: What is the maximal expansion achievable? Could doubly exponential expansion, or even an \emph{unbounded}, expansion of randomness be possible? Can exponential expansion be achieved using a more natural protocol that is robust to noise? What are the minimal assumptions required on the seed quality? While many specific protocols have been considered in the quantum information literature~\cite{Colbeck11,FGS11,colbeck2012free}, to our knowledge no general model of randomness certification and amplification had yet been formulated.

In this paper we introduce a simple and natural framework for randomness amplification which captures nearly all previously considered protocols. We initiate the systematic investigation of the possibilities and limitations of such protocols, which we call \emph{randomness amplifiers}.\footnote{These protocols have been called ``randomness expanders'' or ``randomness expansion protocols'' in prior works, but we adopt the term randomness amplifiers to avoid confusion with the traditional concept of expanders.} In particular, we present both the first \emph{upper bounds} on the achievable randomness expansion of natural protocols, as well as the first robust exponential \emph{lower bounds}. (Note that here, contrary to common usage in theoretical computer science, upper bounds on randomness expansion are \emph{impossibility} results, whereas lower bounds are \emph{possibility} results.) 

\paragraph{A puzzle.} Before describing our results in greater detail, we invite the interested reader to contemplate the following puzzle.
Consider a protocol in which the referee chooses a single pair of uniformly random bits $x,y\in\{0,1\}$, and sends $x^n$ and $y^n$ ($x$ and $y$ repeated $n$ times each) to two non-signaling devices $D_A$ and $D_B$, respectively. The referee collects the devices' output sequences $(a_1,\ldots,a_n)$ and $(b_1,\ldots,b_n)$, and accepts iff $85\% \pm 1\%$ of the rounds $i$ are such that $a_i \oplus b_i = x\wedge y$ (i.e. the $\CHSH$ condition). Under the a priori assumption that the devices pass this protocol with probability at least $99\%$, (i) what is the minimal amount of randomness that the devices must have generated, and (ii) what are strategies for the devices that achieve this while generating as little randomness possible?

Tackling (i) consists in proving a lower bound, while (ii) considers upper bounds. Establishing a lower bound requires ruling out clever ``cheating strategies'' by the devices, in which they would pass the referee's test while still producing outputs with little min-entropy. Upper bounds consist in devising such clever strategies. The upper bounds we prove in Section \ref{sec:upper} demonstrate the possibility for non-trivial cheating strategies, that take advantage of structural properties of the referee's test in order to save on the randomness generated and defeat the protocol.

\paragraph{Robust protocols.} An appealing aspect of randomness amplifiers is that they only rely on two basic physical assumptions: the ability to enforce the non-signaling condition between devices, and the a priori existence of a some small amount of randomness to use as seed. As such, these protocols lend themselves quite naturally to experimental implementations. In fact,~\cite{Pironio} report an implementation of their quadratic randomness amplifier in which $42$ bits of certified randomness were generated (over the course of a month of experiments!). 

However, noise as well as errors due to imperfections in laboratory equipment are unavoidable in such experiments. Given the recent interest in realizing efficient implementations of randomness expansion protocols\footnote{Such protocols have recently been suggested as a benchmark for the closure of the so-called \emph{detection loophole}. We refer to the recent survey~\cite{bellnonlocality} for more details.}, it is important to understand the power and limitations of protocols that behave robustly in the presence of noise and imperfect devices. Some randomness amplifiers, such as the one in \cite{vazirani2012certifiable}, are not robust to noise. Is this an artifact or an intrinsic limitation of protocols that achieve exponential randomness expansion? 

\subsection*{Our results}

\paragraph{The model.}
Our first contribution is the introduction of a natural model for randomness amplifiers. Abstractly, we think of a randomness amplifier as a family of \emph{protocols} describing an interaction between a trusted entity (called the referee) and a pair of black-box devices. The referee selects inputs to the devices, collects outputs, and based on these 
decides to either accept or reject the devices' outputs. 
The protocols are parametrized by a \emph{seed length} $m$, which is the amount of initial randomness required to execute the protocol.
The output of the protocol is defined as the output of the black-box devices over the course of the interaction (provided the referee accepted). The procedure has completeness $c$, soundness $s$, and expansion $g=g(m)$ if (i) there exists a pair of non-signaling devices, called the \emph{ideal devices}
, such that the referee's interaction with them will result in a ``pass'' with probability at least $c$, and (ii) for \emph{any} pair of non-signaling devices (either bound by the laws of quantum mechanics or not, depending on context) such that they pass the protocol with probability at least $s$, the output distribution of the devices has min-entropy at least $g(m)$ --- where, ideally, $g(m)\gg m$.

The interaction between the referee and the devices could a priori be arbitrary. In this paper we restrict our attention to \emph{non-adaptive} protocols. In such protocols the referee uses his random seed to select a pair of input strings to be given to each device. He then provides the inputs one symbol at a time, collecting outputs from the devices. At the end of the interaction, the referee applies a test to the inputs and outputs he has collected. Such protocols are called non-adaptive because the inputs to the devices do not depend on the devices' outputs in previous rounds. Nearly all protocols considered in the literature are non-adaptive.

In addition, we formalize the notion of ``robust'' randomness amplifiers: informally, an amplifier is robust if small deviation from the behavior of the ideal devices still results in acceptance with high probability. Since randomness amplification is based on physical assumptions, it is natural to consider models that are robust to noise or device imperfections. Naturally, allowing noisy devices makes the analysis harder, e.g. to prove lower bounds on robust protocols we have to account for the fact that devices may use the freedom to deviate to cheat the protocol. However, we will also show that in certain cases, non-robust protocols can be cheated by malicious devices that exploit the
possibility for noise-free operation!

Unlike the protocols considered in~\cite{Colbeck09, Pironio, vazirani2012certifiable, PM11, FGS11}, conditioned on passing the protocol, the devices' outputs are only required to have high min-entropy, as opposed to being close to uniform. As alluded to above, the guarantee that the devices' output has high min-entropy allows one to apply a randomness extractor to produce nearly uniform bits -- indeed, that is what these previous works do. However, it is known that randomness extraction for min-entropy sources requires an independent seed of logarithmic length~\cite{radhakrishnan2000bounds}, thus trivially limiting many protocols to exponential expansion! Since our interest is in exploring the limits and possibilities of randomness expansion -- including the possibility of super-exponential expansion -- we make the choice of measuring the output randomness by its min-entropy.

\paragraph{A robust lower bound.} Our first result is a lower bound: we extend and generalize the result of~\cite{vazirani2012certifiable} by devising a randomness amplifier that attains exponential expansion and is robust to noisy devices. The underlying protocol is simple and can be based on any non-local game (and not only the $\CHSH$ game as in~\cite{vazirani2012certifiable}) that is \emph{randomness generating}. Informally, randomness generating games are such that any strategy achieving a success probability strictly higher than the classical value must produce randomized answers, on a certain fixed pair of inputs $(x_0,y_0)$  that depend only on the game, not the strategy. Many examples of games are known to be randomness generating, and we give an additional example based on the Magic Square game~\cite{Arvind:02}.

Fix a two-player game $G$. Let $\eta$ denote the ``noise tolerance'' parameter, $\eps$ a target ``security'' parameter and $R$ a number of rounds. The robust protocol $P_G$ is as follows: in each round, with some small probability $p_c$ the two devices are presented with inputs as prescribed in the game $G$. Such rounds are called game rounds. Otherwise, they are presented with some default inputs $x_0, y_0$ respectively. The referee collects the outputs of the two devices for the $R$ rounds, and checks that on average over the game rounds the devices' inputs and outputs satisfy the game condition a fraction of times that is at least the maximum success probability achievable in $G$ using quantum mechanics, minus $\eta$. 

\medskip
\begin{theorem}[Informal]
Let $m$ be a positive integer. Let $G$ be a randomness generating game, $\eta,\eps>0$ and $P_G$ the protocol described above, for some $R=R(m) \leq \exp(m /\log(1/\eps))$ and $p_c = \Theta(\log(1/\eps)/R)$. Then $P_G$ uses $m$ bits of seed, has completeness $1 - \exp(-\eta^2 R)$, soundness $\eps$ and expansion $g(m)=\Omega(R(m))$. 
\end{theorem}

\paragraph{Upper bounds.} We present the first upper bounds on non-adaptive randomness amplifiers.
Our first upper bound applies to protocols based on \emph{perfect games}, which are games $G$ such that there exists a quantum strategy that wins $G$ with probability $1$ (an example is the Magic Square game described in Appendix~\ref{app:rg}). We consider simple protocols in which the referee's test is to verify that the devices win every single round. We give a simple argument, based on the construction of a ``cheating strategy'' for the devices, showing that any such protocol can achieve at most doubly exponential expansion.

While this simple class of protocols already encompasses some protocols introduced in the literature, such as one described in~\cite{Colbeck09}, many protocols do not use perfect games and such a stringent testing condition from the referee. We thus extend this initial upper bound and show that it also applies to arbitrary non-adaptive randomness amplifiers, provided that they are noise-robust and the ideal devices play each round independently.

\medskip
\begin{theorem}[Informal]
	Let the family of protocols $P=(P_m)$ be a non-adaptive randomness amplifier. Suppose that for all $m\in \N$, $P_m$ is noise-robust and the ideal devices for $P_m$ play each round independently. Then, for all $m\in \N$ there exists two quantum devices that are accepted by the protocol $P_m$ with high probability, but whose output min-entropy is at most $2^{O(2^m)}$. 
\end{theorem}

We refer to Theorem~\ref{thm:double-exp} for a precise statement. The basic idea for the cheating strategy is to show that, provided the referee's seed is short enough, the devices can often deterministically re-use some of their outputs in previous rounds. That the referee's test can be arbitrary complicates the argument somewhat, a priori preventing a systematic re-use by the devices of their past outputs: the test could check for obvious patterns that could arise in any obvious re-use strategies. To get around this we use the probabilistic method to show that for any noise-robust test there exists a randomness-efficient re-use strategy that will fool it. 

Our last upper bound is a stronger, \emph{exponential} upper bound on randomness amplifiers that are based on the $\CHSH$ game and in which the referee's test only depends on the pattern of wins and losses in the game that is observed in the protocol. 
However, our ``cheating strategy'' for such protocols requires the use of perfect non-signaling devices (which are able to win the $\CHSH$ game with probability $1$). As such, the significance of the theorem is in the proof rather than in the statement: it demonstrates the possibility for elaborate cheating strategies that exploit the structure of the protocol in order to be accepted in a highly randomness-efficient way. Overcoming this kind of behavior of the devices is a major hurdle in designing protocols that achieve more than exponential expansion of randomness, a tantalizing open problem that we leave open for further work. 

\paragraph{Related work.} 
As mentioned earlier,~\cite{Pironio}, building on~\cite{Colbeck09}, were the first to obtain a quantitative lower bound on randomness expansion. They showed that quantum or non-signaling devices that demonstrate \emph{any} Bell inequality violation can be used to certify randomness. Fehr et al.~\cite{FGS11} extended this result to demonstrate exponential expansion, although their protocol requires the use of two \emph{unentangled} pairs of devices. Vazirani and Vidick~\cite{vazirani2012certifiable} describe a protocol with exponential expansion that only requires two devices. Their protocol, however, is not robust to noise and is tailored to the specifics of the $\CHSH$ game. 

When considering the use of the bits generated by a randomness amplifier in a cryptographic task it may be necessary to obtain stronger guarantees than simply a lower bound on their min-entropy: indeed, in some cases it is essential that the bits not only appear random by themselves, but are also uncorrelated with any potential adversary (say, the maker of the devices). The protocol of~\cite{FGS11} is proven secure against classical adversaries;~\cite{vazirani2012certifiable} also obtain security against quantum adversaries. In this work we do not consider such extended guarantees of security. 

It is worth noting that the protocols given in~\cite{Colbeck09, Colbeck11} do not formally conform to our model of randomness amplifiers; they are based on the GHZ game, which involves three non-communicating devices. However, our expansion upper bounds can be modified to apply to protocols involving more than two devices (see Appendix~\ref{app:upper} for an example). 

Recent results investigate the use of Bell inequality violations to extract almost uniform bits from a weak random source, without requiring a uniform seed (in contrast with the aforementioned protocols, as well as the protocols discussed in this paper)~\cite{colbeck2012free,gallego2012full}. In particular, these works show that it is possible, using the non-signaling principle as a guarantee, to extract almost uniform randomness from so-called Santha-Vazirani sources. The analogous classical task of deterministically extracting uniform random bits from Santha-Vazirani sources is known to be impossible~\cite{santha1986generating}.  Plesch and Pivoluska~\cite{plesch2013single} extend this result to sources guaranteed to have some amount of min-entropy (which is more general than a Santha-Vazirani source) -- but their protocol requires \emph{three} non-signaling devices. Thinh et al.~\cite{Thinh13minentropy} show limitations on randomness extraction based on Bell inequality violations from general min-entropy sources. 

\paragraph{Acknowledgments.} T.V. thanks Andy Drucker and Avi Wigderson for discussions. M.C. and H.Y. thank Scott Aaronson for his engaging class on Quantum Complexity Theory, for which some of the upper bounds were developed as a course project. H.Y. also thanks Joseph Bebel for helpful comments. We thank the anonymous RANDOM referees as well as Marco Tomamichel for comments that helped improve the presentation of our paper. 

\paragraph{Outline of the paper.} We start with some preliminaries in Section~\ref{sec:prelim}. Our model is introduced in Section~\ref{sec:model}. In Section~\ref{sec:lower} we establish our exponential lower bound, while Section~\ref{sec:upper} contains our doubly exponential and exponential upper bounds.

%% file: prelim.tex
\textbf{Notation}. Given an integer $n$ we write $[n]=\{1,\ldots,n\}$. Given a string $x\in \X^n$, where $\X$ is a finite alphabet, we let $x_{\leq i} = (x_1,\ldots,x_i)$, $x_{>i} = (x_{i+1},\ldots,x_n)$, etc. 
If $\X, \Y$ are alphabets and $\pi$ a probability distribution over $\X\times \Y$, for all $R \in \N$ we let $\pi^{\otimes R}$ denote the product distribution defined over $\X^R \times \Y^R$ by $\pi^{\otimes R}(x_1,\ldots,x_R,y_1,\ldots,y_R) = \prod_{i \in [R]} \pi(x_i,y_i)$.
We use capital letters $X,Y,\ldots$ to denote random variables. Let $X$ be a random variable that takes values in some discrete domain $\mathcal{D}$. Its min-entropy is defined as $\MinEntropy{X} = -\log \max_{x\in\mathcal{D}} \Pr(X = x)$. The Shannon entropy of a random variable $X$ is denoted $H(X)$ as usual. We also define the max-entropy of a random variable $X$ as $H_0(X) = \log(|\mathrm{supp}(X)|)$, where $\mathrm{supp}(X)$ denotes the support of $X$. 
The conditional min-entropy is defined as
$$H_\infty(X|Y) = -\log\Big(\sum_y \Pr(Y=y) 2^{-H_\infty(X|Y=y)}\Big).$$
For two discrete random variables $X, Y$ with the same domain, their statistical distance is $\|X - Y\|_1 = \frac{1}{2}\sum_{x\in \mathcal{D}} |\Pr(X = x) - \Pr(Y = x)|$. For $\eps > 0$, the smoothed min-entropy of a discrete random variable $X$ is defined as 
$$\SmoothMinEntropy{\eps}{X} \,=\, \sup_{\tilde{X},\|\tilde{X} - X\|_1 \leq \eps} \MinEntropy{\tilde{X}},$$
where the supremum is taken over all $\tilde{X}$ defined on $\mathcal{D}$. The smoothed conditional min-entropy is
$$\SmoothMinEntropy{\eps}{X|Y} \,=\, \sup_{(\tilde{X},\tilde{Y}),\|(\tilde{X},\tilde{Y}) - (X,Y)\|_1 \leq \eps} \MinEntropy{\tilde{X}|\tilde{Y}}.$$
We also define the smooth entropy of a random variable $X$, conditioned on an event $T$, as the smooth entropy of a random variable having the distribution of $X$ conditioned on $T$. The following will be useful. 

\begin{claim}\label{claim:cond-min-entropy}
Let $X$ be a random variable, $\eps>0$ and $T$ an event such that $\Pr(T) \geq 1-\delta$. Then $H_\infty^{\eps-2\delta}(X|T) \leq H_\infty^\eps(X)$.
\end{claim}

\begin{proof}
Let $Y$ be a random variable having the same distribution as $X$ conditioned on $T$. 
Let $\tilde{Y}$ be a random variable such that $H_\infty(\tilde{Y}) = H_\infty^{\eps-2\delta}(X|T) = H_\infty^{\eps-2\delta}(Y)$  and $\|\tilde{Y}-Y\|_{1,T} \leq \eps-2\delta$, where both quantities are computed on the probability space conditioned on $T$. Define $\tilde{X}=\tilde{Y}$ on $T$, and extend $\tilde{X}$ to the whole probability space in an arbitrary way, under the condition that $H_\infty(\tilde{X}) \geq H_\infty(\tilde{Y})$. Then $\|\tilde{X}-X\|_1 \leq (\eps-2\delta)/(1-\delta)+\delta \leq \eps$, proving the claim.  
\end{proof}

We will make use of the following basic relations between the different entropy measures.
\begin{lemma}
\label{lem:basic_entropy}
	Let $X$ be a discrete random variable over some domain $\X$. Let $\eps \in [0,1)$. Then,
	\begin{enumerate}
		\item $H_{\infty}(X) \leq H(X) \leq H_0(X)$, and
		\item $\SmoothMinEntropy{\eps}{X} \leq H_0(X) - \log(1 - \eps)$.
	\end{enumerate}
\end{lemma}

\begin{proof}
	The first inequality in the first item follows because $H(X)$ is a convex combination of $\{-\log(\Pr(X = x))\}$ values over all $x\in \mathrm{supp}(X)$, and $H_{\infty}(X)$ is the minimum such value. A proof of the inequality $H(X) \leq H_0(X)$ can be found in~\cite[Ch.\ 2]{cover2012elements}. We prove the second item. Let $U = \mathrm{supp}(X)$. Let $\mu = \SmoothMinEntropy{\eps}{X}$, and let $Y$ be a random variable such that $\|Y - X\|_1 \leq \eps$ and $\MinEntropy{Y} = \mu$. Then, for every $u \in U$, $\Pr(Y = u) \leq 2^{-\mu}$ by definition, but we also must have $|U| \cdot 2^{-\mu} \geq 1 - \eps$ because of the statistical distance between $Y$ and $X$. Since $H_0(X) = \log |U|$, the proposition follows.
\end{proof}

We will use some standard concentration bounds (see e.g. Chapter~1 in~\cite{Dubhashi98concentrationof} for a detailed introduction). 

\begin{fact}[Chernoff bound]\label{fact:chernoff}
Let $X_1,\ldots,X_n$ be independent Bernoulli random variables with expectation $p$. Then 
$$\Pr\left[ \Big|\sum_{i=1}^n X_i -pn \Big| \geq \delta pn \right] \,\leq\, 2\,e^{-\delta^2pn/3}.$$
\end{fact}

\begin{fact}[Hoeffding's inequality]\label{fact:hoeffding} Let $X_1,\ldots,X_n$ be independent centered random variables such that for every $i\in [n]$, we have $\Pr(X_i \in [a_i,b_i]\big)=1$ . Then for any $t\geq 0$, 
$$\Pr\left[ \Big|\sum_{i=1}^n X_i  \Big| \geq t \right] \,\leq\, 2\,e^{-2t^2/\sum_i (b_i-a_i)^2}.$$
\end{fact}

\paragraph{Two-player games.} A two-player game $G$ is specified by input alphabets $\mathcal{X}$ and $\mathcal{Y}$, output alphabets $\mathcal{A}$ and $\mathcal{B}$, an input distribution $\pi$ on $\mathcal{X}\times\mathcal{Y}$, and a game predicate $G:\X\times \Y\times \A \times\B \to \{0,1\}$. The game is played between a referee and two non-communicating players, who we typically call Alice and Bob. The referee generates inputs $x\in\mathcal{X}$ and $y\in\mathcal{Y}$ according to $\pi$, and sends them to Alice and Bob respectively. Alice answers with $a\in\mathcal{A}$ and Bob answers with $b\in\mathcal{B}$. The referee accepts iff $G(a,b,x,y)=1$, in which case we say that the players win (or pass) the game. 

\paragraph{Strategies.} Given a game $G$, we define its \emph{value} as the maximum winning probability of two players in the game, where the probability is taken over the referee's choice of inputs and randomness that may be part of the players' strategy. In full generality, a strategy $S$ is specified by a family of distributions $\{p_S(\cdot,\cdot|x,y):\A\times\B\to [0,1]\}_{(x,y)\in\X\times\Y}$, parametrized by input pairs $(x,y)$ and defined over the output alphabet $\A\times \B$. The value of $G$ clearly depends on restrictions that we may place on the allowed families of distributions, and we (as is customary in the study of two-player games in the quantum literature) consider three distinct restrictions: 

First, if the players are restricted to classical deterministic strategies, specified by functions $f_A:\X\to\A$ for Alice and $f_B:\Y\to\B$ for Bob, we obtain the \emph{classical value}, which is defined as 
$$\omega_c(G)\,=\, \max_{f_A,f_B} \sum_{x,y} \pi(x,y)  G(f_A(x),f_B(y),x,y).$$
It is not hard to see that the use of private or even shared randomness by the players will not increase the classical value. Second, by allowing all strategies that may be implemented locally using quantum mechanics, including the use of entanglement, one obtains the \emph{quantum value} of $G$, $\valq(G)$. In this paper we will not need to use the formalism of quantum strategies, and we refer to e.g.~\cite{CHTW04} for a good introduction. Finally, we may allow any strategy which respects the non-signaling principle: the only restriction on the players' family of distributions is that it satisfies
$$ \forall x\in\X,y,y'\in\Y,a\in\A,\qquad p_S(a|x,y)\,=\,\sum_b p_S(a,b\mid x,y)\,=\,\sum_b p_S(a,b\mid x,y')\,=\,p_S(a|x,y'),$$
and a symmetric condition holds when marginalizing over the first players' output. The corresponding value is called the \emph{non-signaling} value $\valns(G)$. It is clear that, for any game $G$, $\omega_c(G)\leq \omega_q(G)\leq \omega_{ns}(G)$. Examples of games are known for which all three inequalities are strict (the $\CHSH$ game, see below). There are also games for which the first inequality is strict, and the second is an equality (the Magic Square game, see below), and for which the first inequality is an equality and the second is strict (see e.g.~\cite{PhysRevLett.99.180502}).

\paragraph{The CHSH game.}
The $\CHSH$ game is a two-player game with two non-communicating players, Alice and Bob, who are given independent random inputs $x,y\in \{0,1\}$ respectively. Their task is to produce outputs $a,b\in\{0,1\}$ such that $a \oplus b = x \wedge y$. By enumerating over all deterministic strategies, it is not hard to see that $\omega_c(\CHSH) = 3/4$. There is a simple quantum strategy based on the use of a single EPR pair that demonstrates $\omega_q(\CHSH) \geq \cos^2(\pi/8) \approx 85\%$, and in fact it is an optimal quantum strategy~\cite{cirel1980quantum,nielsen2010quantum}. Furthermore, $\omega_{ns}(G)=1$ (see Lemma~\ref{nosignalCHSH} for the simple proof). Thus, the $\CHSH$ game is an example of a game $G$ such that $\omega_c(G) < \omega_q(G) < \omega_{ns}(G)$.

\paragraph{(Non-adaptive) Protocols.} Informally, a protocol prescribes the interaction between a trusted \emph{referee} and a pair of \emph{devices}, which we usually denote by $D_A$ and $D_B$. A protocol can be thought of as a multi-round game in which the rounds are played sequentially; we use the word ``devices'' rather than ``players'' to refer to the fact that the interaction may go on for many rounds, but there is no essential difference.  
In this paper, we restrict our attention to \emph{non-adaptive} protocols, where the referee's messages to the devices are independent of the devices' outputs. Formally, a non-adaptive protocol $P$ is specified by a tuple $\langle \X, \Y, \A, \B, R, \pi, T \rangle$, where: $\X, \Y$ are finite input alphabets, $\A, \B$ are finite output alphabets, $R\in \N$ is the number of rounds of interaction, $\pi$ is the input probability distribution over $\X^R \times \Y^R$, and $T:\X^R \times \Y^R \times \A^R \times \B^R \to \{0,1\}$ is the referee's \emph{test}. 

Given such a protocol $P$, the interaction between the referee and a pair of devices $(D_A, D_B)$ proceeds as follows: using private randomness, the referee samples the input sequence $(x,y)\in \X^R\times \Y^R$ from $\pi$. Then, for each round $i\in [R]$, the referee distributes $x_i\in \X$ and $y_i\in \Y$ to $D_A$ and $D_B$, respectively. Devices $D_A$ and $D_B$ are required to produce outputs $a_i\in \A$ and $b_i\in \B$, respectively. Let $a = (a_i)$ and $b = (b_i)$. After $R$ rounds of interaction, the referee \emph{accepts} if $T(x,y,a,b) = 1$. Otherwise, the referee \emph{rejects}.

Given a protocol $P$ and a pair of devices $(D_A, D_B)$, a \emph{strategy} for the devices is a description of their behavior in the protocol: for each round index $i$, a family of distributions $\{p(a_i,b_i | x_i,y_i,\text{hist}_i)\}$ on $\A_i\times \B_i$, where $\text{hist}_i$ is the \emph{history} of the protocol prior to round $i$, i.e. the list of inputs and outputs generated by the devices in previous rounds. We call a strategy quantum (resp. non-signaling) if it can be implemented using isolated quantum (resp. non-signaling) devices. 

%% file: model.tex
In this section we define the notion of \emph{randomness amplifiers} that we use throughout the paper. A randomness amplifier is given by a family $(P_m)_{m\in \N}$ of non-adaptive protocols. The following definition summarizes the important parameters associated with a non-adaptive randomness amplifier.

\medskip
\begin{definition} A family of protocols $P = (P_m)=\langle \X, \Y, \A, \B, R_m, \pi_m, T_m \rangle$ is a \textbf{randomness amplifier} with \textbf{seed length} $m$, \textbf{completeness} $c=c(m)$, \textbf{soundness} $s=s(m) < c$ against quantum (resp. non-signaling) strategies, \textbf{smoothness} $\eps=\eps(m)$, \textbf{expansion} $g=g(m)$ and \textbf{ideal strategy} $S_{\text{ideal}}=S_{\text{ideal}}(m)$ if the following hold for every $m\in\N$:
\begin{itemize}
	\item (Seed length) A sequence of inputs $(x,y)\in\X^{R_m}\times\Y^{R_m}$ to the devices can be sampled according to $\pi_m$ using at most $m$ uniformly random bits,
	\item (Completeness) If the devices behave as prescribed in the ideal strategy $S_{\text{ideal}}$,\footnote{We refer to devices implementing the ideal strategy as \emph{ideal devices}.}  then
	\begin{equation}\label{eq:honest}
	\Pr(T_m(X,Y,A,B)=1) \geq c(m), 
	\end{equation}
	where $A$, $B$ are random variables corresponding to each device's outputs, and the probability is over $(X,Y)\sim \pi_m$ and the randomness inherent in the strategy. 
	\item (Soundness) For all quantum (resp. non-signaling) strategies $S$ for the devices in $P_m$, if playing according to $S$ guarantees $\Pr(T_m(X,Y,A,B) = 1) \geq s(m)$, then
	$$ \SmoothMinEntropy{\eps}{A,B \mid T_m(X,Y,A,B) = 1} \geq g(m). $$
\end{itemize}
\end{definition}

For notational clarity we will often omit the parameter $m$ when the seed length is clear from context.

We further elaborate on the completeness and soundness conditions. We say that the completeness of a randomness amplifier $P$ holds \emph{with quantum (resp. non-signaling) devices} whenever the ideal strategy can be implemented using quantum (resp. non-signaling) devices. Similarly, we say that the soundness of $P$ holds \emph{against quantum (resp. non-signaling) devices} if the universal quantifier in the soundness condition is over all quantum (resp. non-signaling) strategies. Generally, a stronger condition on the soundness (i.e. soundness against non-signaling devices) will imply weaker parameters, such as smaller expansion. 

We note that the amount of randomness produced is measured according to its ($\eps$-smooth) min-entropy. Motivation for this particular measure comes from the fact that it tightly characterizes the number of ($\eps$-close to) uniform bits that can be extracted from the devices' outputs using a procedure known as an extractor (we refer to~\cite{Ren05} for more details on using extractors for privacy amplification, including in the quantum setting). This procedure requires the use of an additional short seed of uniformly random bits, which we do not take into account here: our goal is simply to produce \emph{entropy}, and one could in principle replace the min-entropy by, say, the Shannon entropy in the definition. We also observe that the conditioning on $T(X,Y,A,B)$ in the definition of the soundness is necessary. 
Indeed, consider devices applying the following strategy: first, flip a biased coin that is heads with probability $p$. If the coin comes up heads, deterministically output $0^R$ (a strategy which we may assume will fail the referee's test with high probability over his choice of inputs). Otherwise, apply the ideal strategy specified in the protocol. The probability that the devices pass the protocol is at least $(1-p)c$ (where $c$ is the completeness parameter of the protocol), which is larger than the soundness $s$ as long as $p<1-s/c$, a value larger than $1/2$ for any reasonable setting of $s$ and $c$. For any $\eps<p$ the $\eps$-smooth min-entropy of the device's outputs is at most $\log (1/p)$; it is (potentially) large only once one conditions on success. 

It may be useful to keep typical ranges for the different parameters in mind. The ``asymptotic'' quantity is the seed length $m$. Completeness will often be exponentially close to $1$ in the number of rounds $R$, itself a function of $m$ that can range from linear to doubly exponential (or more). The soundness and smoothness will be exponentially small in $m$.

We now define restricted classes of protocols which capture most of the protocols so far introduced in the literature. The definitions are extended to randomness amplifiers in the natural way. 

\textbf{Natural protocols.} We will say that a protocol $P$ is \emph{natural} if there is a two-player game $G$ such that the ideal strategy for $P$ is the strategy $S_G^{\otimes R}$ consisting of playing each of the $R$ rounds of $P$ according to an optimal (quantum or non-signaling depending on the context) strategy $S_G$ for the game $G$. We say that $G$ is the game that \emph{underlies} $P$. All randomness amplifiers to date are natural according to this definition. In this paper we only consider natural protocols. 

\begin{definition}\label{def:prod_test}
Let $G$ be a two-player game. A test function $T:\X^R\times \Y^R \times \A^R \times \B^R \to \{0,1\}$ is a \textbf{product test} with respect to $G$ iff there exists a function $g:\{0,1\}^R \to \{0,1\}$ such that $T(x,y,a,b) = g \left ( G(x_1,y_1,a_1,b_1),\ldots,G(x_R,y_R,a_R,b_R)  \right)$.
\end{definition}

\textbf{Product protocols.} We will say that a  
protocol $P$ is a \emph{product protocol} if the referee's test $T$ is a product test with respect to some two-player game $G$.  Intuitively, the protocol $P$ consists of $R$ independent instances of the game $G$, played in sequence (though the input distribution may not necessarily be the product distribution $\pi_G^{\otimes R}$). The referee's test is to apply a function $g$ on the sequence of wins and losses of the devices. Natural examples of functions $g$ for this purpose include the $\mathsf{AND}$ function and threshold functions, e.g. $g(w) = 1$ iff the Hamming weight of $w\in\{0,1\}^R$ is greater than $(\omega_q(G) - \eta)R$. An example of a \emph{non}-product test would be one where, say, the referee checks that the devices output $(0,0)$ (for a given input pair) in $\frac{1}{2} \pm \epsilon$ fraction of the rounds.

\textbf{Robust protocols.} Informally, a protocol is robust if small deviations from an ideal strategy are still accepted with high probability by the referee. We now provide a formal definition for such protocols. First, we introduce the notion of closeness of strategies. Let $P$ be an $R$-round protocol. Let $X,Y$ be random variables on $\X^R$, $\Y^R$ respectively distributed according to the protocol's input distribution $\pi_P$. For any strategy $S$, let $S_i(X_{\leq i},Y_{\leq i})$ denote the random variable distributed as the devices' outputs in round $i$, conditioned on having played according to $S$ on the input sequence $(X_{\leq i},Y_{\leq i})$. Then we say that two strategies $S$ and $\widehat{S}$ are $\eta$-close if for all rounds $i\in [R]$, 
$$
	\big\|S_i(X_{\leq i},Y_{\leq i}) - \widehat{S}_i(X_{\leq i},Y_{\leq i})\big\|_1 \leq \eta.
$$

Let $P$ be a protocol with some specified ideal strategy $S_{\mathrm{ideal}}$ that is accepted with probability at least $c$ in the protocol (as is when $P$ is a member of a randomness amplifier, for example). Let $T$ be the referee's test in the protocol. We say that $P$ is $\eta$-\emph{robust} if whenever the devices' strategy $S$ for the protocol $P$ is $\eta$-close to $S_{\text{ideal}}$, it holds that $\Pr(T(X,Y,A,B) = 1) \geq c$ (under strategy $S$). We note that this definition captures the concept of robustness against not only, say, i.i.d. noise, but also against physically plausible sources of imperfection such as misaligned mirrors, imperfect detectors, etc.

%% file: lower_bounds.tex
Let $G$ be a two-player game in which inputs to Alice (resp. Bob) are chosen from sets $\mathcal{X}$ (resp. $\mathcal{Y}$), and answers expected in sets $\mathcal{A}$ (resp. $\mathcal{B}$).  Let $\pi$ be the referee's distribution on input pairs in $G$.

\medskip
\begin{definition}\label{def:rg}
We say that a two-player game $G$ is $(p_0,\eta,1-\xi)$-randomness generating against quantum (resp. non-signaling) players if there exists an input $x_0\in\mathcal{X}$ such that the marginal probability $\pi(x_0)\geq p_0$ and any quantum (resp. non-signaling) strategy for the players that has success at least $\omega_q(G)-\eta$ (resp. $\omega_{ns}(G)-\eta$) satisfies
\begin{equation}\label{eq:rg}
 \max_{a\in\mathcal{A}} p\big(A=a \mid X=x_0\big) \leq 1-\xi.
\end{equation}
\end{definition}

We note that for any given game $G$, $x_0$ and $\eta>0$ the problem of approximating the smallest possible $\xi$ such that $G$ is $(\pi(x_0),\eta,\xi)$-randomness generating against quantum (resp. non-signaling) devices is an optimization problem for which upper bounds can be obtained through a hierarchy of semidefinite programs~\cite{Doherty08,Pironio} (resp. a linear program). If $G$ is an XOR game, the hierarchy converges at the first level: there is an exact semidefinite program of size polynomial in $|\mathcal{X}||\mathcal{Y}|$. For the special case of the $\CHSH$ game, choosing $x_0=0$ it is known that $\CHSH$ is $(1/2,\eta,1/2+\sqrt{3\eta})$-randomness generating (see Claim~\ref{claim:chsh-rg}). In Claim~\ref{claim:ms-rg} we show that the Magic Square game is $(1/9,\eta,12/13+\eta)$-randomness generating. Clearly, the condition that $\eta< \omega_q(G)-\omega_c(G)$ (resp. $\eta< \omega_{ns}(G)-\omega_c(G)$) is necessary for the game $G$ to be randomness generating for any $\xi> 0$.

\subsection{Unbounded randomness expansion}

For any game $G$ with input distribution $\pi$, $\eps>0$ and function $R:\N \to \N$, we introduce a simple randomness amplifier that achieves unbounded expansion, with the strong limitation that soundness only holds against devices that are restricted to play each round of the protocol in a completely isolated, though not necessarily identical, manner (in particular, the devices are memory-less but may be aware of the round number). Fix an optimal strategy $S$ for $G$. Our randomness amplifier is given by the family of protocols $(P_m)$, where protocol $P_m$ is defined as follows.

$P_m$ has $R = R(m)$ rounds. The rounds are divided into $(1/\eps)$ blocks $B_j$ of $\eps\, R$ rounds each. For each block, the referee chooses a random pair of inputs $(x,y) \sim \pi$ that is used in every round of the block. The referee then checks that in every block at least a $\omega_q(G|S,x,y)-\eta$ fraction of the rounds have been won, where here $\omega_q(G|S,x,y)$ is defined as the probability that the players satisfy the game condition, conditioned on their inputs being $(x,y)$, in the fixed strategy $S$ (so that $\sum_{x,y}\pi(x,y)\omega_q(G|S,x,y) = \omega_q(G)$). (In the non-signaling case, replace $\omega_q$ by $\omega_{ns}$.) The referee accepts the devices if and only if this condition holds in every block. Note that $P$ is a non-adaptive protocol with ideal strategy $S^{\otimes R}$, completeness that goes exponentially fast to $1$ with $R$, and seed length $O(\eps^{-1})$ (where we treat the size of $G$ as a constant). 

The following lemma shows that the randomness amplifier $(P_m)$ has good soundness and a constant rate. Since the seed length remains a constant as $R(m)$ grows, the protocol can be used to achieve unbounded expansion. 

\begin{lemma}\label{lem:infinite}
Let $\eta,\xi>0$ and $G$ a $(p_0,4\eta,1-\xi)$-randomness generating game against quantum (resp. non-signaling) players. Then, for all $\eps>0$ and functions $R:\N \to \N$ the above-described randomness amplifier $(P_m)$ has 
\begin{enumerate}
	\item Seed length $O(\eps^{-1})$,
	\item Completeness $1-e^{-\Omega(\eps R(m))}$ with quantum (resp. non-signaling) devices, 
	\item Soundness $e^{-\Omega(1/\eps)}$ against \emph{independent} quantum (resp. non-signaling) devices,
	\item Smoothness $e^{-\Omega(1/\eps)}$, and
	\item Expansion $g(m) = \alpha R(m)$, where $\alpha$ is a positive constant depending only on $\xi$ and $\eta$. 
\end{enumerate}
Furthermore, $P$ is $\eta$-robust.
\end{lemma}

\begin{proof}
The argument is simple and makes heavy use of the independence assumption; we only sketch it here. We do the proof in the quantum case; the non-signaling setting is similar. For each round $i$ and pair of inputs $(x,y)$ let $p_i(x,y)$ be the $i$-th round devices' success probability in game $G$, when the inputs are deterministically fixed to $(x,y)$. Consider a fixed block $B_j\subseteq [R]$ of $\eps R$ rounds, and suppose that in that block it holds that 
\begin{equation}\label{eq:inf-1}
\frac{1}{\eps R}\sum_{i\in B_j}\, \textsc{E}[p_i(x,y)] \,\leq\, \omega_q(G)-2\eta,
\end{equation}
 where the expectation is taken according to the input distribution $\pi$ in game $G$. Then there must exist a pair of inputs $(x^j,y^j)$ such that $(1/(\eps R))\sum_{i\in B_j} p_i(x^j,y^j) \leq \omega_q(G|S,x^j,y^j)-2\eta$. For any $i\in [R]$ let $(X_i,Y_i)$ be random variables denoting the referee's choice of inputs to the devices in round $i$, and $Z_i$ a binary random variable that is $1$ if and only if the game is won in round $i$. By definition $\textsc{E}[Z_i\mid (X_i,Y_i)=(x^j,y^j)]=p_i(x^j,y^j)$. Applying Hoeffding's inequality (see Fact~\ref{fact:hoeffding}), conditioned on the input to block $B_j$ being chosen as $(x^j,y^j)$ it holds that
$$ \Pr\Big( \frac{1}{\eps R}\sum_{i\in B_j} Z_i \geq \omega_q(G)-\eta \Big) \,\leq\, e^{-\Omega(\eta^2 \eps R)}.$$
Let $f = \min_{(x,y):\pi(x,y)>0} \pi(x,y)$; $f$ is a constant depending only on $G$. In any block $B_j$ the probability that the input to the block is $(x^j,y^j)$ is at least $f$. Since the inputs to different blocks are chosen independently, applying Hoeffding's inequality once more the probability that less than a fraction $f/2$ of blocks $B_j$ have their input set to $(x^j,y^j)$ is at most $e^{-\Omega(f^2/\eps)}$. As a result, except with probability exponentially small in $1/\eps$ a constant fraction of the blocks $B_j$ are such that~\eqref{eq:inf-1} does not hold. In particular, at least half of rounds $i$ in any such block must be such that $p_i := \textsc{E}[p_i(x,y)] \geq \omega_q(G)-4\eta$, where we used $p_i\leq \omega_q(G)$, by definition of the optimum $\omega_q$. Using the definition of a randomness generating game, provided the input in that round is $(x_0,y_0)$ --- which happens with constant probability ---  the outputs produced by the devices in that round must contain a constant amount of entropy.
\end{proof}

\subsection{Exponential randomness expansion}

It is much more realistic to assume that the devices \emph{do} have memory, and we analyze this case for the remainder of the section. For any game $G$ that is randomness generating we show that there exists a corresponding randomness amplifier with exponential expansion. For simplicity we only consider quantum strategies; the non-signaling setting is completely analogous. We introduce a randomness amplifier $(P_m)$ which is parametrized by a randomness generating game $G$, a fixed set of inputs $(x_0,y_0)\in\mathcal{X}\times\mathcal{Y}$, an error tolerance $\eta>0$, a precision $\eps=\eps(m)$, a ``checking probability'' $p_c=p_c(m)$ and a number of rounds $R=R(m)$.

Fix an $m\in \N$. We first describe the input distribution in $P_m$. Let $w_{max}=2p_cR$, and $\mathcal{U} \subseteq \{0,1\}^R$ the set of binary strings with Hamming weight at most $w_{max}$. Let $q$ be the distribution on $\{0,1\}^R$ with density $q(x)=\prod_{i\in [R]} p_c^{x_i}(1-p_c)^{1-x_i}$. Let $\delta$ be a precision parameter and $q_\delta$ be defined on $\{0,1\}^R$ by $q_\delta(x) = (\delta/ R^{w_{max}}) \lfloor q(x)( R^{w_{max}}/\delta) \rfloor$ if $x\in\mathcal{U}$ and $q_\delta(x)=0$ otherwise. Clearly $\|q_\delta\|_1\leq 1$; normalize $q_\delta$ by introducing an additional ``fail'' symbol $\perp$ such that $q_\delta(\perp)=1-\sum_{x\in\mathcal{U}} q_\delta(x)$.  We think of $q_\delta$ as a discretized version of $q$; the following claim will be useful. 

\begin{claim}\label{claim:sample-bin}
Assume $\delta > 2e^{-p_c R/3}$. Then $q_\delta$ is supported on $\mathcal{U} = \{x\in\{0,1\}^R:\, |x|\leq 2p_cR\}$, $\|q-q_\delta\|_1\leq 2\delta$, and it is possible to sample from $q_\delta$ using $O(p_c R \log(R))$ uniformly random bits. 
\end{claim}

\begin{proof}
By definition, for any $x\in\mathcal{U}$, $|q(x)-q_\delta(x)|\leq \delta /|\mathcal{U}|$, where we used $|\mathcal{U}|\leq R^{w_{max}}$. Using the Chernoff bound (Fact~\ref{fact:chernoff}), under $q$ it holds that $\Pr_q(x\notin\mathcal{U}) \leq 2e^{-p_c R/3}<\delta$. Overall, $\|q-q_\delta\|_1\leq 2\delta$ as claimed. 
To sample from $q_\delta$, first sample a weight $w\in \{0,\ldots, w_{max}\}$. Using the discretized form of $q_\delta$ this can be done using $O(p_cR+\log(R^{w_{max}}/\delta))$ bits. Then sample a uniformly random string of weight $w$ using at most $w\log(R)$ bits.  
\end{proof}

The protocol $P_m$ proceeds as follows. The referee first samples a string $u\in\{0,1\}^R$ distributed according to $q_{\eps^2/4}$. He then selects inputs for the devices in the $R$ rounds. If $u_i=1$ inputs are selected as prescribed in $G$; such rounds are called ``game rounds''. If $u_i=0$ they are set to the default value $(x_0,y_0)$. Once inputs to the $R$ rounds have been computed, the referee sequentially provides them to the devices, who produce a corresponding sequence of outputs. The referee computes the average number of rounds in which the input/output pairs satisfy the game condition $G$, and accepts if and only if it is at least $\omega_q(G)-\eta$. We note that $P_m$ is a \emph{natural}, \emph{product} protocol for which we define the ideal strategy to consist of playing each round independently according to an optimal quantum strategy for the game $G$. With that ideal strategy, the protocol is also $\eta$-robust. 

The following theorem shows that for any game $G$ that is $(p_0,\eta,1-\xi)$-randomness generating against quantum adversaries,\footnote{For simplicity we focus here on establishing completeness and soundness for quantum devices, but our arguments can easily be extended to the non-signaling case.} for some $\xi>0$, the protocols $(P_m)$ form a randomness amplifier with exponential expansion. 

\begin{theorem}
Let $G$ be $(p_0,4\eta/p_0,1-\xi)$-randomness generating against quantum players, with input distribution $\pi$. Let $m_\pi$ be the number of uniform random bits required to sample a pair of inputs $(x,y)\sim\pi$. Let $p_c,R,\eps,s:\N\to\N$ be non-negative functions such that $ p_c(m) R(m)(\log R(m)+m_\pi)\leq m/C$, $\eps(m) \leq s(m)$, and $s(m) \eps(m) > e^{-C \min(\eta^2,p_0\xi^2) p_c(m) R(m)}$ for all $m$, where $C$ is a universal constant. Then the family of protocols $(P_m)$ (as defined above), based on game $G$, inputs $(x_0,y_0)$, error tolerance $(p_0\eta/4)$, precision $\eps$, checking probability $p_c$ and number of rounds $R$ is a randomness amplifier with
\begin{enumerate}
	\item Seed length $m$,
	\item Completeness $c\geq 1-e^{-\eta^2 R(m)}$ with quantum devices,
	\item Soundness $s$ against quantum devices,
	\item Smoothness $\eps$, and
	\item Expansion $g(m)\geq \xi R(m)/5$. 
\end{enumerate}
Furthermore, $(P_m)$ is $\delta$-robust for any $\delta<p_0\eta/4$.
\end{theorem}

For any small constant $\eta>0$, integer $m$ and desired soundness and smoothness $\eps=s$, setting $R(m) = C'm/\log(1/\eps)$ and $p_c = C''\log(1/\eps)/R$ for small enough $C'$ and large enough $C''$ (depending on $\eta$, $p_0$ and $\xi$) will lead to parameters that satisfy the theorems' assumptions, thus guaranteeing an amount of min-entropy generated that is exponential in $m$ for constant $\eps$. 

The claim on the completeness in the theorem follows by a standard Chernoff bound. The claim on the seed length follows immediately from the description of the referee given above and the bound in Claim~\ref{claim:sample-bin}. Finally, the claims on the soundness, smoothness and rate follow from Proposition~\ref{prop:main-robust} below, which shows that if the claims are not satisfied, then there exists a strategy for the players in the game $G$ that contradicts the assumption that $G$ is $(p_0,4\eta/p_0,1-\xi)$-randomness generating (to see this, set the only new parameter $\delta$ in the proposition to $\delta = \xi/5$).

\begin{proposition}\label{prop:main-robust}
Let $1/2\geq \delta\geq 2p_c$, $\eta>0$ and $s\geq \eps>0$ be such that 
$$\frac{\log(16/(\eps^2 s))}{R} \,<\,\frac{\min(p_0\delta^2,\eta^2) p_c}{30},$$
and suppose further that $H_\infty^\eps(A,B\mid X,Y,T(A,B,X,Y)=1) < \delta R$ and $\Pr(T(A,B,X,Y)=1)\geq s$. Suppose that 

Then there exists a single-round pair of quantum devices and an $a_0\in\mathcal{A}$ such that when the game $G$ is played with the devices it holds that 
$$\Pr(G(A,B,X,Y)=1)\geq \omega_q(G)-4\eta/p_0\qquad\text{and}\qquad \Pr(A=a_0\mid X=x_0)\geq 1-5\delta,$$
where $p_0=\pi(x_0)$ is the marginal probability that input $x_0$ is chosen for Alice in the game $G$.  
\end{proposition}

\begin{proof}
To prove the proposition we analyze a slightly different protocol, in which the referee's procedure is replaced by the following simpler one: for each round $i\in [R]$, set $u_i=1$ independently with probability $p_c$, and define $w:=\sum_i u_i$. Then proceed as prescribed in the description of protocol $P_m$ above to choose inputs to the devices. By Claim~\ref{claim:sample-bin}, the statistical distance between the distribution on inputs chosen by this simplified referee and the original one is at most $\eps^2/2$. Hence the distribution of outputs produced by the same devices under the one or the other referee's input distribution will also have statistical distance at most $\eps^2/2$; conditioning on the event that $T(A,B,X,Y)=1$, which has probability at least $\eps$, will at most increase this distance to $\eps/2$. It will thus suffice to prove the proposition for the simplified referee under the restricted assumption that $H_\infty^{\eps/2}(A,B\mid X,Y,T(A,B,X,Y)=1) < \delta R$ to deduce the proposition for the original referee. 

Let $\Omega = \{(x,y,a,b,u)\in (\{0,1\}^5)^R\}$ be the probability space associated with the experiment consisting of executing the protocol with the devices. Here $(x,y)$ are the strings of inputs chosen by the (simplified) referee, $(a,b)$ the outputs observed, and $u$ a string of bits that indicates the locations chosen for the game rounds (which correspond to $u_i=1$). For every $i\in [R]$ let $U_i\in\{0,1\}$ be the random variable that is $1$ if and only if $u_i=1$. Let $W=\sum_i U_i$. By definition, $T(A,B,X,Y)=1$ if and only if 
$$\frac{1}{W}\sum_{i:U_i=1} 1_{G(X_i,Y_i,A_i,B_i)=1} \geq \omega_q(G)-\eta.$$
Applying the Chernoff bound (Fact~\ref{fact:chernoff}), since each round is chosen as a game round independently with probability $p_c$, 
$$\Pr\Big(\big|W-p_cR\big|\,\geq\, \frac{p_cR}{3}\Big)\,\leq\, 2e^{-p_c R/27}\,\leq\, \eps^2/4,$$
where the second inequality follows from our choice of parameters. Furthermore, if $w'$ is the number of rounds such that $w_i=1$ and $x_i=x_0$, and $W'$ the associated random variable, then similarly
$$\Pr\Big(\big|W'-p_cp_0R\big|\,\geq\, \frac{p_cp_0R}{3}\Big)\,\leq\, 2e^{-p_cp_0 R/27}\,\leq\, \eps^2/4.$$
Define an event $\win$ as the event that $T(A,B,X,Y)=1$ and 
\begin{equation}\label{eq:main-robust-0}
\big|W-p_cR\big|\,\leq\, \frac{p_cR}{3},\qquad \big|W'-p_cp_0R\big|\,\leq\, \frac{p_cp_0R}{3}.
\end{equation}
Further conditioning on $\win$, the assumptions of the proposition together with Claim~\ref{claim:cond-min-entropy} and $\eps^2\leq \eps/2$ imply that $H_\infty^{\eps/2}(A,B\mid \win,X,Y) < \delta R$ and $\Pr(\win)\geq s/2$, where we used $\eps\leq s$. By definition, the first condition implies that for any distribution $q$ such that $\|q-p\|_1\leq \eps/2$ (where here $q,p$ are taken as distributions on the probability space $\Omega_W$ obtained from $\Omega$ by conditioning on $\win$), $H_\infty(A,B\mid \win,X,Y) <\delta R$, where here the min-entropy is taken with respect to the distribution $q$. In particular, it must be that the set $S$ of all $(x,y,a,b,u)\in \win$ such that $\Pr((A,B)=(a,b)\mid (X,Y)=(x,y)) > 2^{-\delta R}$ has probability at least 
\begin{equation}\label{eq:prs}
\Pr(S)\,=\,\Pr(S\mid \win)\Pr(\win)\,\geq\, (\eps/2) (s/2)\,=\,s\eps/4.
\end{equation}
The following two claims show properties of those sequences $(x,y,a,b,u)\in S$. 

\begin{claim}\label{claim:robust-pf-1} 
For all but a fraction at most $\eps$ of all $(x,y,a,b,u)\in S$ it holds that 
\begin{equation}\label{eq:det}
\frac{1}{w'} \sum_{i\in [R],\,u_i=1,\,x_i=x_0} \Pr\big(A_i = a_i \mid X_i = x_0,\,(A,B,X,Y)_{<i}=(a,b,x,y)_{<i}\big) \geq  1-4\delta.
\end{equation}
\end{claim}

\begin{proof}
Let $(x,y,a,b,u)\in S$. By definition, $\Pr((A,B)=(a,b)\mid (X,Y)=(x,y)) > 2^{-\delta R}$. Applying Bayes' rule and taking logarithms we get
$$ \sum_{i=1}^R -\log\Pr(A_i = a_i \mid (X,Y)_i = (x,y)_i,\, (A,B,X,Y)_{<i}=(a,b,x,y)_{<i}) < \delta R,$$
where we used that $A_i$ is independent of $(X,Y)_{>i}$. Using concavity of the logarithm, we get
$$ \frac{1}{R}\sum_{i=1}^R \Pr(A_i = a_i \mid (X,Y)_i = (x,y)_i,\, (A,B,X,Y)_{<i}=(a,b,x,y)_{<i}) > 2^{-\delta}\geq 1-\delta.$$
Note that $A_i$ is independent of $Y_i$. Moreover, since $(x,y,a,b,u)\in \win$ there are at most $4p_cR/3$ game rounds, hence at least $(1-4p_c/3)R$ rounds must have $x_i=x_0$. Therefore, 
$$\frac{1}{R}\sum_{i=1}^R \Pr(A_i = a_i \mid X_i=x_0,\, (A,B,X,Y)_{<i}=(a,b,x,y)_{<i}) \geq 1-\delta-4p_c/3 \geq 1-2\delta.$$
Finally, note that conditioned on $X_i=x_0$ (and $(A,B,X,Y)_{<i}=(a,b,x,y)_{<i}$) any given round $i$ is chosen as a game round independently with probability $p_cp_0/(1-p_c+p_cp_0)\geq p_cp_0/2$; the distribution of $A_i$, conditioned on $X_i=x_0$, does not depend on this choice. Applying Hoeffding's inequality (Fact~\ref{fact:hoeffding}), 
$$\Pr\Big(\frac{1}{W'}\sum_{i\in [R], U_i=1} \Pr(A_i = a_i \mid X_i=x_0,\, (A,B,X,Y)_{<i}=(a,b,x,y)_{<i}) \leq 1-4\delta\Big) \leq 2e^{-8\delta^2 W'/3} \leq s\eps^2/4,$$
where here the summation is restricted to those rounds in which $U_i=1$ and $X_i=x_0$, and for the second inequality we used the bound on $W'$ and our choice of parameters. Using the lower bound~\eqref{eq:prs} on the size of $S$, the claim is proved.
\end{proof}

\begin{claim}\label{claim:robust-pf-2}
For all but a fraction at most $\eps$ of $(x,y,a,b,u)\in S$ it holds that 
\begin{equation}\label{eq:win}
\frac{1}{w} \sum_{i\in [R],\,u_i=1} \Pr\big(G(X_{i},Y_{i},A_{i},B_{i})=1 \,\big|\, (A,B,X,Y)_{<{i}} = (a,b,x,y)_{<{i}}\big) \geq \omega_q(G)-2\eta.
\end{equation} 
\end{claim}

\begin{proof}
For any $j=1,\ldots,W$ define a random variable $Z_j\in\{0,1\}$ on $\Omega_W$ by $Z_j=1$ if and only if $G(x_{i_j},y_{i_j},a_{i_j},b_{i_j})=1$, where $i_j$ is the index of the $j$-th game round. By definition of $\win$, it holds that 
\begin{equation}\label{eq:main-robust-1}
\sum_j \,Z_j  \,\geq\, W(\omega_q(G)-\eta).
\end{equation}
For any $k=1,\ldots,W$ let $V_k = \sum_{j=1}^k (Z_j - \textsc{E}[Z_j \mid Z_{j-1},\ldots,Z_{1},U])$.
Then $(V_k)$ is a martingale with respect to the filtration induced by the sequence of random variables $(W,Z_1),(W,Z_1,Z_2),\ldots,(W,Z_1,\ldots,Z_W)$. Applying Azuma's inequality (see e.g. Theorem~5.2 in~\cite{Dubhashi98concentrationof}),
\begin{equation}\label{eq:main-robust-2}
\Pr\Big( \Big|\sum_j Z_j - \sum_j \textsc{E}\big[Z_j|Z_{j-1},\ldots,Z_1,W\big] \Big| \geq W\eta \Big)\,\leq\, 2\,e^{-W\eta^2/2}.
\end{equation}
Using~\eqref{eq:main-robust-0},~\eqref{eq:main-robust-2} together with~\eqref{eq:main-robust-1} implies that
\begin{equation}\label{eq:main-robust-3}
\Pr\Big( \sum_j \textsc{E}\big[Z_j\mid Z_{j-1},\ldots,Z_1,W\big] \,\leq\, W(\omega_q(G)-2\eta) \Big)\,\leq\, \eps^2\, s/4,
\end{equation}
given our choice of parameters. The probability here is taken over $\Omega_W$; removing the conditioning on $\win$ will give a probability over $\Omega$ that is at most $\eps\, s/8$. On $\Omega$, for any $j$ it holds by definition that $\textsc{E}[Z_j\mid Z_{j-1},\ldots,Z_1,W\big]\leq \omega_q(G)$. Using that $\Pr(S)\geq s\eps/4$, Eq.~\eqref{eq:main-robust-3} implies that all but a fraction at most $\eps$ of $(x,y,a,b,u)\in S$ are such that~\eqref{eq:win} holds. 
\end{proof}

Using Claims~\ref{claim:robust-pf-1} and~\ref{claim:robust-pf-2} we may now conclude the proof of the proposition. Fix any $(x,y,a,b,u)\in S$ such that both~\eqref{eq:det} and~\eqref{eq:win} hold. By an averaging argument a round $i$ such that $u_i=1$ and $x_i=x_0$ can be found such that both equations hold with the ``loss'' on the right-hand side multiplied by $W'/W\leq 2p_0$ for the case of~\eqref{eq:win} and any constant greater than $1$ for~\eqref{eq:det}. Fix such an $i$. 
 Execute the protocol with the devices up to the $i$-th round (excluded), choosing inputs as prescribed by $(x,y)$. If the outputs produced by the devices do not match $(a,b)$ in every round, abort and restart. Conditions~\eqref{eq:det} and~\eqref{eq:win} guarantee that, once the conditioning succeeds, the two devices at the beginning of round $i$ will be in a state such that both conditions stated in the conclusion of the proposition hold.
\end{proof}

%% file: upper_bounds.tex
In this section we prove upper bounds on the expansion attainable by a wide class of randomness amplifiers. The upper bounds are proved by exhibiting ``cheating strategies'' for the two devices $D_A$ 
and $D_B$ that fool a referee into accepting, while producing an amount of entropy that is at most doubly exponential in the referee's seed length. In particular, our bounds on output entropy are independent of the number of rounds.

The main idea behind the cheating strategies we exhibit is that, after a sufficiently large number of rounds, there are inevitable correlations between the referee's inputs to the devices that hold irrespective of the referee's choice of random seed. These correlations can be inferred from the given input distribution $\pi$ of the protocol, before it begins. In Theorems~\ref{thm:easydouble} and~\ref{thm:double-exp} we use the observation that after a number of rounds that is doubly exponential in the referee's seed length, the inputs to $D_A$ and $D_B$ in the current round $i$ must be identical to their inputs in some previous round $j < i$. If the referee's test is particularly simple (as it is assumed to be in Theorem~\ref{thm:easydouble}), then the devices can pass the protocol by simply copying their answers from round $j$. More generally, we show that for robust protocols there will be a set of rounds $J \subseteq [R]$ such that $|J| = 2^{O(2^m)}$ (where $m$ is the referee's seed length), and a strategy for the devices to deterministically recombine their respective answers from the rounds in $J$ into answers for the rounds in $[R]\backslash J$. It follows that the devices' output entropy is at most $O(|J|) = 2^{O(2^m)}$.

An important element of the cheating strategies we present is the \textbf{input matrix}, which is defined for any nonadaptive protocol as follows.

\begin{definition}[Input matrix]\label{def:input-matrix}
Let $P$ be an $R$-round, non-adaptive protocol with seed length $m$. The input matrix $M_P$ is the $R \times 2^m$ matrix whose $(i,\sigma)$-entry is $M_P(i,\sigma)=(X(\sigma)_i,Y(\sigma)_i)$, where here $X(\sigma)$ (resp. $Y(\sigma)$) are the input sequences for device $D_A$ (resp. $D_B$) chosen by the referee on seed $\sigma\in\{0,1\}^m$.
\end{definition}

When $P$ is clear from context we shall simply write $M$ instead of $M_P$ for the input matrix. We let $M_i \in (\X\times \Y)^{2^m}$ denote the $i$th row of an input matrix $M = M_P$. We define the set $F(M) \subseteq [R]$ as the set of round indices $i$ such that $i\in F(M)$ iff $M_i \neq M_j$ for all $j < i$. The following immediate claim places a bound on the size of $F(M)$. 

\begin{claim}\label{claim:fm}
Let $P$ be a protocol with seed length $m$ and input alphabets $\X,\Y$. Then $|F(M)| \leq |\X \times \Y|^{2^m}$.
\end{claim}

\subsection{A simple doubly exponential bound}\label{sec:easy-ub}

We first demonstrate a doubly exponential upper bound on randomness amplifiers that are based on perfect games, which are games $G$ such that $\omega_q(G) = 1$ (or $\valns(G) = 1$, if we're allowing devices with full non-signaling power). In these protocols, the referee checks that the devices win every single round.

\medskip
\begin{theorem}  \label{thm:easydouble} 
Let $G$ be such that $\valq(G) = 1$ (resp. $\valns(G) = 1$). Let $P=(P_m)$ be a randomness amplifier with input (resp. output) alphabets $\X,\Y$ (resp. $\A,\B$) and in which the referee's test consists in verifying that the devices win $G$ in every round. Suppose completeness and soundness of $P$ both hold with quantum (resp. non-signaling) devices. Then the expansion of $P$ satisfies
$$g(m)\,\leq\, |\X \times \Y|^{2^m}  \log|\A\times \B| - \log(1 - \eps(m)),$$
where $\eps(m)$ is the smoothness of $P$. 
\end{theorem}

We only sketch the proof here; we give a more general argument in the next section. The idea of the proof is as follows: in each round $i$, the devices check whether $i\in F(M)$ or not, where $M = M_{P_m}$ is the input matrix corresponding to protocol $P_m$. If it is, then the devices play according to the ideal, honest strategy that wins $G$ with probability $1$. If not, then there must exist a $j\in F(M)$, $j < i$, such that $M_i = M_j$. Thus, regardless of the referee's seed, it must be that  $(x_i,y_i)=(x_j,y_j)$ always. In that case, the devices will simply replay their outputs $(a_j,b_j)$ from that round, independently setting $a_i:=a_j$ and $b_i:=b_j$. Since we can assume that round $j$ was won with probability $1$, round $i$ must be won with probability $1$ as well. It is easy to see that the only entropy-generating rounds are those in $F(M)$, and the theorem follows from Claim~\ref{claim:fm} and Lemma~\ref{lem:basic_entropy}.

\subsection{A doubly exponential bound for robust protocols}

In this section we generalize the bound from the previous section to show a doubly exponential upper bound on the expansion achievable by any randomness amplifier based on a protocol that is non-adaptive and robust. In particular, the underlying game $G$ may not be perfectly winnable, and the referee's test $T$ may not necessarily check that the devices win $G$ in every single round. The fact that we allow an arbitrary test $T$ in the protocol complicates the proof, as the referee may now for example check for obvious answer repetitions in the players' answers to identical question pairs, and thereby easily detect cheating strategies of the form described in Section~\ref{sec:easy-ub}. Nevertheless, we will design a somewhat more elaborate cheating strategy for the devices in any such protocol, that prevents it from achieving unbounded expansion. 

\begin{theorem}\label{thm:double-exp}
Let $P=(P_m)$ be a natural, $\eta$-robust randomness amplifier such that completeness and soundness both hold with respect to quantum (resp. non-signaling) devices. Let $K_m = \Omega\left(\frac{1}{\eta^2} \log \frac{|\A \times \B|\cdot |F(M_{P_m})|}{\eta} \right)$. Then the expansion of $P$ satisfies 
$$g(m) \,\leq\, K_m \cdot  |F(M_{P_m})| \cdot \log|\A\times\B|-\log(1-\eps(m)),$$
where $\A,\B$ are the output alphabets of $P$, and $\eps(m)$ is the smoothness of $P$.
\end{theorem}

Combined with Claim~\ref{claim:fm}, the theorem implies that any $\eta$-robust randomness amplifier $P$ must have a expansion $g(m) = 2^{O(2^m)}$ (where the constant in the $O(\cdot)$ depends only on $\eta$, the smoothness $\eps$, and the alphabets $\X, \Y, \A, \B$). This in particular demonstrates that unbounded randomness expansion as demonstrated in Lemma~\ref{lem:infinite} is impossible as soon as the devices are allowed to have (classical) memory.

The idea for the proof is simple. Instead of directly reusing outputs corresponding to identical pairs of inputs, as described in Section~\ref{sec:easy-ub}, the devices first repeatedly apply the protocol's ideal quantum (resp. non-signaling) strategy for game $G$ in order to locally generate a discrete approximation to the corresponding distribution on outputs. Whenever they receive a pair of questions for which they already computed such an approximation, they use shared randomness to jointly sample a pair of answers from the approximating distribution. To conclude we use the probabilistic method to derandomize the shared sampling step (which would otherwise still lead to the generation of a constant amount of entropy per round). 

\begin{proof}[Proof of Theorem~\ref{thm:double-exp}]
Fix an $m$ and protocol $P_m$, with $R=R_m$ rounds. Consider the following randomness-inefficient strategy $S'$ for the devices. Since $P_m$ is a natural protocol, it has has an ideal strategy of the form $S_G^{\otimes R}$, for $S_G$ is a single-round two-player quantum strategy for $P_m$'s underlying game $G$. Let $M=M_{P_m}$ be the input matrix for protocol $P_m$. At every round $i$, $D_A$ and $D_B$ locally check whether $i\in F(M)$. If so, they first perform the following \textbf{sampling step}: repeatedly apply the strategy $S_G$ a number $K= \Omega\left(\frac{1}{\eta^2} \ln \frac{|\A \times \B|\cdot |F(M)|}{\eta} \right)$ times on their respective inputs $x_i$ and $y_i$. Let the outcomes of the $K$ instances be $a^{(i)}=(a^{(i)}_k)_{k=1,\ldots,K}$ and $b^{(i)}=(b^{(i)}_k)_{k=1,\ldots,K}$. Each device stores its own sequence of outcomes. 
Whether or not $i\in F(M)$, the devices then perform the following \textbf{replay step}. They identify the unique $j\leq i$ such that $j\in F(M)$ and $M_j=M_i$. Using shared randomness they select a uniformly random $k\in [K]$, and output $a_k^{(j)}$ and $b_k^{(j)}$ respectively.

Define the following probability density function on $\A\times\B$: for all $i\in F(M)$, for all $(a,b)\in \A\times\B$, 
\begin{equation}\label{eq:double-exp-1}
q_i(a,b) = \frac{1}{K} \sum_{k=1}^K 1_{\big(a^{(i)}_k,b^{(i)}_k\big) = (a,b)}.
\end{equation}
Assume first that the strategy $S'$ achieves winning probability $\Pr (T(X,Y,A,B) = 1) \geq c$. Let $V$ denote the devices' shared classical randomness, as it is used in the replay steps. By averaging, there exists a fixed setting $V^*$ such that the probability that $T(X,Y,A,B) = 1$, when using $V^*$, is at least $c$. Let $S$ be the strategy where Alice and Bob perform the sampling steps as usual, but in the replay steps, they use the fixed string $V^*$ instead (which they can precompute beforehand). Thus, the entropy of the outputs produced by the strategy $S$ comes entirely from the sampling steps. There are at most $|F(M)|$ sampling steps, and in each step, at most $K \log|\A\times\B|$ bits of randomness are produced, so $H_0(A,B \mid T(A,B,X,Y) = 1) \leq |F(M)| \cdot K \cdot \log|\A\times\B|$. We use Lemma~\ref{lem:basic_entropy}, and the theorem follows provided we can show that $S'$ achieves the desired success probability whenever $K$ is chosen as stated. 

To show this, we use the assumption that $P_m$ is an $\eta$-robust protocol. From the definition of $\eta$-robust and the strategy $S'$, it will suffice to verify that with high probability for every $i\in F(M)$ the distribution with density $q_i$, as defined in ~\eqref{eq:double-exp-1}, is $\eta$-close in statistical distance to the distribution implied by $S_G$, for the pair of inputs $(x_i,y_i)$. This follows from a standard application of Hoeffding's inequality: for any fixed $i$, $\eta>0$ and $(x_i,y_i)$ the probability that $\|q_i - S_G(\cdot,\cdot|x_i,y_i)\|_1 > \eta/2$ is at most $|\A \times \B| \cdot \exp(-O(\eta^2 K))$. By the union bound, the probability that there exists an $i\in F(M)$ such that $\|q_i - S_G(\cdot,\cdot|x_i,y_i)\|_1 > \eta$ is at most $|\A \times \B|\cdot |F(M)|\cdot \exp(-O(\eta^2 K))$. By our setting of $K$, this probability can be  made less than $\eta/2$.
\end{proof}

\subsection{An exponential upper bound for protocols with non-signaling devices}

In this section we prove exponential upper bounds on the attainable expansion of a class of non-adaptive randomness amplifiers for which completeness holds with respect to non-signaling devices.  We address protocols using the $\CHSH$ game, which have been widely studied in the literature~\cite{Pironio,vazirani2012certifiable}. 

\begin{theorem} \label{CHSHAND}
Let $P = (P_m)$ be a randomness amplifier in which completeness and soundness both hold with non-signaling devices, and for each $m$ the referee's test $T_m$ is a product test with respect to the $\CHSH$ game.  Then
$$g(m) \leq  2^{2m+2} - \log(1 - \eps(m)),$$
where $\eps(m)$ is the smoothness parameter of $P$.
\end{theorem}

Theorem~\ref{CHSHAND} exhibits a scenario in which the specific structure of the underlying game $G$ and the protocol can be used to give an exponential improvement over Theorem~\ref{thm:easydouble}. For simplicity we have constrained the theorem statement to protocols involving the $\CHSH$ game, but the proof can be extended to establish the same result when $G$ is a balanced 2-player XOR game, as well as the (3-player) GHZ game, which has played an important role in early randomness expansion results~\cite{Colbeck11}. We refer to Appendix~\ref{app:upper} for additional details.  

We remark that Theorem~\ref{CHSHAND} implies a ``meta-theorem'' that says that the type of analysis performed in~\cite{vazirani2012certifiable} cannot be improved to have more than exponential expansion. Any randomness amplifier based on the $\CHSH$ game in which the referee only checks that the devices won more than a certain fraction of the rounds, and where the analysis of soundness only uses the fact that the devices are non-signaling, by Theorem~\ref{CHSHAND}, must be limited to exponential expansion. The randomness amplifier in~\cite{vazirani2012certifiable} is of this form, and hence modifying it to obtain super-exponential expansion would require either a non-product test, \emph{or} an analysis that uses the fact that the devices can ``only'' be quantum! 

\begin{proof}[Proof of Theorem~\ref{CHSHAND}]
Fix an integer $m$ and protocol $P=P_m$, with test $T$ and number of rounds $R$. 
Let $G(x,y,a,b) = 1 \oplus x y \oplus a \oplus b$ (i.e. the $\CHSH$ game predicate). For simplicity we first prove the theorem in the special case when the product test is
$$T(x,y,a,b) =  \prod_{i=1}^R G(x_i,y_i,a_i,b_i).$$ 
We give a strategy that can be used by the non-signaling devices $D_A$ and $D_B$ to ensure that $T(X,Y,A,B) = 1$ with probability 1. The strategy will have the additional property that all of the output pairs $(a_i, b_i)$, except for at most $2^m$ values of $i$, are deterministic functions of the outputs produced (using the ``honest'' strategy described in the proof of Lemma~\ref{nosignalCHSH}) in a particular set of $2^{2m}$ previous rounds.  This proves the desired result.

Let $M$ be the protocol's input matrix, as introduced in Definition~\ref{def:input-matrix}. Let us consider the rows of the input matrix $M$ as vectors $M_i \in \mathbb{F}_2^{2^{m+1}}$.  
Additionally, before the protocol begins, the devices precompute the set $I \subseteq [R]$, which consists of all $i$ such that $M_i$ (as a vector in $\mathbb{F}_2^{2^{m+1}}$) is linearly independent from $\{ M_j : j < i \}$. Note that $|I| \leq 2^{m+1}$.
 
We now describe the strategy employed by $D_A$ and $D_B$. In each round $i$, $D_A$ and $D_B$ check whether $i\in I$. If so, they perform the \textbf{sampling step}. Otherwise, they perform the \textbf{replay step}. Let $X_i$ and $Y_i$ denote the inputs to $D_A$ and $D_B$ in the $i$th round.

\textbf{Sampling step.} Let $i$ be a round in which $D_A$ and $D_B$ perform the sampling step. Let $I(i) = \{j \in I : j \leq i\}$. $D_A$ and $D_B$ play two series of private $\CHSH$ games, $S_1 = (C_{ij})$ and $S_2 = (C_{ji})$ for all $j\in I(i)$, and store the outcomes without reporting them to the referee. Using a canonical ordering of these games (e.g. playing series $S_1$ first, where the $C_{ij}$ are played in order of increasing $j$, and then $S_2$, where $C_{ji}$ are played in order of increasing $j$), the devices $D_A$ and $D_B$ use the perfect non-signaling strategy described in Lemma~\ref{nosignalCHSH} to play $C_{ij}$, and obtain outputs $A_{ij}$ and $B_{ij}$, respectively. Similarly, they will play the games $C_{ji}$ and obtain outputs $A_{ji}$ and $B_{ji}$, respectively.
Since we are using the perfect non-signaling strategy, for all $j\in I(i)$, we have $G(X_i,Y_j, A_{ij},B_{ij}) = G(X_j,Y_i,A_{ji},B_{ji}) = 1$. Note that the devices can play this series of private games without communicating.

Finally, $D_A$ and $D_B$ report outputs $A_i = A_{ii}$ and $B_i = B_{ii}$ to the referee.

\textbf{Replay step.} If $D_A$ and $D_B$ perform the replay step in round $i$, we have that $M_i$ is linearly dependent on the rows $\{ M_j : j < i \}$. Observe that the set $\{ M_j : j \in I(i) \}$ forms a linearly independent basis over $\mathbb{F}_2^{2^{m+1}}$ for the rows $\{ M_j : j \leq i \}$.  Thus, there exists a subset $J \subset I(i)$ such that $M_i  = \sum_{j \in J} M_j$, and it follows that, regardless of the value of random seed chosen by the referee, $(X_i, Y_i) = \sum_{j \in J} (X_j, Y_j) = (\sum_{j \in J} X_j, \sum_{j \in J} Y_j)$.  Knowing this, $D_A$ and $D_B$ now wish to produce output values $A_i$ and $B_i$ respectively (without communicating), such that $G(X_i,Y_i, A_i,B_i) =  1 \oplus X_i Y_i \oplus A_i \oplus B_i = 1$, which is equivalent to
$$A_i \oplus B_i =  X_i Y_i = \left (\sum_{j \in J} X_j \right )\left( \sum_{j \in J} Y_j \right) = \sum_{(k,j) \in J^2} X_k Y_j.$$
To accomplish this, $D_A$ outputs $A_i = \sum_{(k,j) \in J^2} A_{kj}$ and $D_B$ outputs $B_i = \sum_{(k,j) \in J^2} B_{kj}$, where the values of the summands are the outputs generated in the sampling steps described above. By design, for each $(k,j) \in J^2 \subset I(i)^2$, we have $A_{kj} \oplus B_{kj} = X_k  Y_j$.
It follows that
\begin{align*}
& A_i \oplus B_i = \sum_{(k,j) \in J^2} A_{kj} \oplus \sum_{(k,j) \in J^2} B_{kj} = \sum_{(k,j) \in J^2} A_{kj} \oplus B_{kj} \\ 
&  =  \sum_{(k,j) \in J^2} X_k Y_j = \left (\sum_{j \in J} X_j \right) \left(  \sum_{j \in J} Y_j \right) =  X_i Y_i
\end{align*}
which implies that $G(X_i,Y_i, A_i,B_i) =  1$, as desired. Thus, for every round $i \in [R]$, we have that $G(X_i,Y_i,A_i,B_i) = 1$ with probability $1$, and hence $T(X,Y,A,B) = 1$ with probability $1$. 

We now show the upper bound on the entropy of the devices' outputs. In every round, the outputs in all steps are a deterministic function of the round number and the set of outputs $\{ A_{ij}, B_{ij} : (i,j) \in I^2 \}$. Since this set contains exactly $|I|^2$ random variables, each of which has max-entropy $1$, the entire set can have max-entropy at most $|I|^2$. Thus $H_{\text{max}}(A, B) \leq |I|^2$. From our previous bound on $|I|$,  we have $H_{\text{max}}(A,B) \leq |I|^2 \leq 2^{2m+2}$. The upper bound on the smooth min-entropy follows from Lemma~\ref{lem:basic_entropy}.

This concludes the proof in the case that $T(X,Y,A,B) = \prod_{i=1}^R G(X_i,Y_i,A_i,B_i)$. We now indicate how the proof can be extended to general product tests $T$.  

As we saw above, $D_A$ and $D_B$ have a non-signaling strategy that allows them to pass each individual $\CHSH$ test with probability $1$, and produce at most $2^{2m+2}$ bits of entropy in their outputs.  We now want a similar proof which allows $D_A$ and $D_B$ to win against any $\CHSH$ product test, where an arbitrary function $g$ is used to combine the outcomes of the tests performed in each round.   Suppose that the test is specified by   
$$T(X,Y,A,B) = g \left ( G(X_1,Y_1,A_1,B_1),\ldots, G(X_R,Y_R,A_R,B_R)  \right)$$
for some function $g: \{0,1 \}^R \to  \{0,1\} $.  Since $c>0$ we know that the referee cannot reject every vector of wins and losses, so there must exist some $v \in \{0,1 \}^R$ such that $g(v) =1$. We can think of $v$ as specifying a sequence of $\CHSH$ wins and losses. $D_B$ can fix such a $v$ before the start of the protocol. $D_A$ and $D_B$ will perform exactly the same strategy as above, except where $D_B$ would have output $B_i$ in the $i$th round, $D_B$ will now output $B_i \oplus v_i \oplus 1$. It is easy to see that $G(X_i,Y_i,A_i,B_i \oplus v_i \oplus 1) = v_i$. Thus $T(X,Y,A,B \oplus v \oplus 1) = g(v) = 1$, and $D_A$ and $D_B$ will pass the referee's test with probability $1$. We again have $H_{\text{max}}(A, B) \leq 2^{2m+2}$, and the desired result follows. 
\end{proof}

We note that the cheating strategy exhibited in the proof of Theorem~\ref{CHSHAND} crucially relies on the existence of \emph{noiseless} devices. As such, the theorem suggests an intriguing possibility: that the assumption of an unavoidable presence of noise in any devices used to execute a given protocol may allow for the certification of \emph{additional} randomness, by ruling out special finely-tuned adversarial strategies.

%% file: appendices.tex
\section{A non-signaling strategy for CHSH}

We note that there is a no-signaling strategy that succeeds in the $\CHSH$ game with probability $1$. An analogue of Claim~\ref{claim:chsh-rg} also holds against no-signaling strategies (see~\cite{Pironio} Appendix A.3).

\medskip
\begin{lemma}  \label{nosignalCHSH}
There exists a non-signaling strategy that wins the $\CHSH$ game with probability 1. 
\end{lemma}

The proof of Lemma \ref{nosignalCHSH} is well-known, but may be instructive for readers unfamiliar with non-signaling strategies.

\begin{proof}
Labeling the inputs to the game as $x$ and $y$ respectively, imagine that the outputs ($a$ and $b$, resp.) are selected according to the following conditional distribution.

If $x \wedge y = 1$ then the two possible outputs pairs are $(a,b) = (1,0)$, and $(a,b) = (0,1)$ each with occurring probability $\frac{1}{2}$.  If $x \wedge y = 0$ then the output pairs are $(a,b) = (0,0)$, and $(a,b) = (1,1)$, again each occurring with probability $\frac{1}{2}$.  It now follows easily that, regardless of the values of $a$, $b$, $x$, and $y$ we have  
$$\sum_{b'} p(a,b' \mid x,y) = p(a \mid x,y) = p(a \mid x) = \frac{1}{2}$$
and
$$\sum_{a'}  p(a',b \mid x,y) = p(b \mid x,y) = p(b \mid y) = \frac{1}{2}$$
Thus, the above strategy is non-signaling by definition, and wins with probability 1.

\end{proof}

\section{Some randomness generating games}\label{app:rg}
\begin{claim}\label{claim:chsh-rg}
For any $\eta\geq 0$ the	game $\CHSH$ is $(1/2,\eta,f(\eta))$-randomness generating (against quantum strategies) for $f(\eta) = \frac{1}{2} + \sqrt{3\eta}$.
\end{claim}

\begin{proof}
	Consider a quantum strategy for the $\CHSH$ game whose  success probability is at least $\Pr(\win) \geq \valq(\CHSH) - \eta$, where $\valq(\CHSH) = \cos^2(\pi/8)$. 
	It is proved in~\cite{Pironio} that for every $a$ and $x$ in $\{0,1\}$,
$$
	\Pr(A = a \mid X = x) \leq \frac{1}{2} \left(1 + \sqrt{2 - I^2/4} \right),
$$
where $I = 8\Pr(\win) - 4 = 8(\valq(\CHSH) - \eta) - 4$ is the so-called ``Bell correlation value''. Observe that $\valq(\CHSH) = (2 + \sqrt{2})/4$ for $\CHSH$, so
\begin{align*}
	\frac{1}{2}\left(1 + \sqrt{2 - I^2/4} \right) &\leq \frac{1}{2} + \sqrt{2} \cdot \sqrt{\sqrt{2} \eta - 2\eta^2} \leq \frac{1}{2} + \sqrt{3\eta}.
\end{align*}
\end{proof}

\textbf{The Magic Square game}. Consider a $3\times 3$ matrix, and suppose that one is asked to fill in each entry with $1$ or $0$, with the constraint that each row must have even parity and each column must have odd parity. Clearly, there is no such assignment that satisfies all the constraints, because while the row constraints imply that the sum of the entries has even parity, the column constraints imply that the same sum has odd parity, a contradiction.

Now consider the following $2$ player game, which we call the $\magic$ game. The referee chooses an $x\in [6]$ uniformly at random, interpreted as choosing a row or column of a $3\times 3$ matrix at random. Then, the referee chooses a $y\in [3]\times[3]$ that corresponds to a random entry in the row/column $x$. For example, conditioned on $x = 1$, $y$ is uniform over the set $\{(1,1), (1,2), (1,3) \}$, the entries in the first row. 
The referee sends $x$ to Alice, and solicits Alice for an assignment $a\in\{0,1\}^3$ to the entries in that row/column. Simultaneously, the referee sends $y$ to Bob and solicits Bob for an assignment $b\in\{0,1\}$ to entry $y$.  The referee checks that Alice's answer satisfies the parity constraint, and Alice's answer is consistent with Bob's.

From the foregoing discussion, it is easy to see that there is no classical strategy for Alice and Bob to successfully pass the referee's test with probability $1$; in fact it is not hard to show that $\omega_c(\magic)=17/18$. However, there \emph{is} a quantum strategy for Alice and Bob to win with probability $1$~\cite{Arvind:02}: $\omega_q(\magic)=\omega_{ns}(\magic)=1$.

To show that $\magic$ is randomness generating, we derive a contradiction by transforming any near-deterministic strategy for the players into a strategy for the guessing game, which is defined as follows: Alice and Bob receive inputs $x$ and $y$ from the Magic Square input distribution, respectively, and they win the guessing game if Alice outputs $y$. Clearly, there is no non-signaling strategy for Alice that allows her to guess Bob's output with probability greater than $1/3$.

\medskip
\begin{claim}\label{claim:ms-rg}
	Let $\eta < 1/13$. The game $\magic$ is $(1/9,\eta,f(\eta))$-randomness generating (against both quantum and no-signaling strategies) for $f(\eta) = 12/13 + \eta$.
\end{claim}
	
\begin{proof}
Suppose for contradiction that for all $y$, $\max_b \Pr(B = b \mid Y = y) > 12/13 + \eta$. We show that this cannot happen, as it gives rise to a strategy for a guessing game in which Alice guesses Bob's input $y\in[3]\times[3]$ with probability better than $1/3$, which is impossible.
	
	Let $S$ be the strategy employed by Alice and Bob to win the $\magic$ game with probability $1 - \eta$, and such that for every $y$ there exists an output $b^*(y) \in \{0,1\}$ for Bob such that $\Pr(B = b^*(y) \mid Y = y) > 12/13 + \eta$. The function $b^*(y)$ defines an assignment to the $3\times 3$ matrix. There must exist a row or column that violates the parity constraint. Without loss of generality, say that it is the first row.
	
	We now describe the strategy for the guessing game. On input $x$, Alice acts according to strategy $S$ on $x$ and records her output as $a = (a_1,a_2,a_3)$. On input $y$, Bob acts according to strategy $S$ on $y$ (and doesn't need to record any output). If $x$ is not the first row, Alice randomly selects one of the three possible coordinates from the row or column denoted by $x$, and outputs this as her guess for $y$. Otherwise, suppose $x$ is the first row. If Alice's output $(a_1,a_2,a_3)$ doesn't satisfy the parity constraint, she aborts the protocol. The number of $a_i$ that agree with $b^*(1,i)$ is either $0$ or $2$; if it is $0$, Alice will abort the protocol. Otherwise, Alice randomly selects from the two coordinates in agreement and produces that as her guess.
	
	In case that $x$ is not the first row, Alice guesses Bob's input with probability $1/3$. If it is, conditioned on winning the protocol and Bob outputting $b^*(y)$, Alice guesses Bob's input with probability $1/2$. Therefore, 
\begin{align*}
	\Pr(\text{Alice guesses $y$}) &\geq \Pr[\text{Alice guesses $y$} \mid B = b^*(y), \win]\cdot \Pr[B = b^*, \win] \\
																&> \left(\frac{\Pr[\text{$x$ is not first row}]}{3} + \frac{\Pr[\text{$x$ is first row}]}{2} \right) (1 - \eta - (1/13 - \eta)) \\
																&= 1/3
\end{align*}	
where $\Pr[B = b^*, \win] \geq 1 - (1 - \Pr[\win]) - (1 - \Pr[B = b^*]) > 1 - \eta - (1/13 - \eta)$ by the union bound.
\end{proof}

\section{Extending Theorem~\ref{CHSHAND} to arbitrary XOR games}\label{app:upper}

Here we will briefly discuss the extension of Theorem~\ref{CHSHAND} to 2-player XOR games, as well as the GHZ game.  

While considering 2-player XOR games we will, for simplicity, restrict our attention to games which have exactly one valid answer parity for each pair of inputs.  We refer to such games as balanced games.  All balanced games have the form $G(X,Y,A,B) = f(X,Y)\oplus A \oplus B$, where $f(X,Y) = c_1 \oplus c_2 X \oplus c_3 Y \oplus c_4 XY$ for some constants $c_1,c_2, c_3, c_4 \in \{0,1\}$.  The constant term and linear terms can be removed by making a classical addendum to the quantum strategy.  For example, by having Alice XOR her answer with $c_1 \oplus c_2 X$, and Bob XOR his answer with $c_3 Y$.  In this way we can reduce without loss of generality to the case $f(X,Y) = c_4 XY$.  If $c_4 = 0$ then we are done, if $c_4 = 1$ then we have the CHSH game, for which we already know the correct strategy.  So, the proof for balanced 2-player XOR games is an easy extension of that for CHSH, because, in some sense, CHSH characterizes the only interesting example of a 2-player XOR game in this context.

In the (3-player) GHZ game, the three devices are each given an input, which we'll call $X$, $Y$, and $Z$ respectively.  Further, they are guaranteed that $X \oplus Y \oplus Z = 0$.  Their goal is to produce outputs ($A$, $B$ and $C$ respectively) such that $G(X,Y,Z,A,B,C) = f(X,Y,Z)\oplus A \oplus B \oplus C = 0$, where 
$$f(X,Y,Z) \equiv X \vee Y \vee Z = X\oplus Y\oplus Z\oplus XY\oplus YZ\oplus XZ\oplus XYZ$$

We note that for GHZ $\valq(G) = 1$ (this is well known, see \cite{Colbeck11}).  Thus, the three devices can win this game with probability 1 using a quantum strategy.  The analog of Theorem~\ref{CHSHAND} for the GHZ game can be obtained by following the proof of Theorem~\ref{CHSHAND} with slight modifications that we will now describe.  The linear terms of $f(X,Y,Z)$ can be dealt with by a classical modification of the strategy in which the first player XOR's their answer ($A$) with $X$, the second player XOR's their answer with $Y$, etc.  The $XY$, $YZ$, and $XZ$ terms can be dealt with by secretly using the probability 1 non-signalling strategy for CHSH between the respective devices so that the CHSH game essentially is used as a subroutine in the cheating strategy.  For example, if the first and second players secretly play CHSH on using inputs $X'$ and $Y'$ (resp.), and a probability 1 strategy, then they obtain outputs $A'$ and $B'$ respectively such that $A'\oplus B' = X'Y'$.  By XORing these outputs onto their final output they effectively remove the $X'Y'$ term from $f(X,Y,Z)$.  In the case that the input $(X,Y,Z)$ is a linear combination of previous inputs, we can extend this method to deal with the quadratic number of quadratic cross terms, in the same manner as in the cheating strategy for CHSH protocols.  This same technique is used between all three pairs of players.

Lastly, the $XYZ$ term can be dealt with (for any particular input $X_iY_jZ_k$) by secretly playing a series of GHZ games with those inputs.  Since the test $G$ only uses the XOR of the three outputs, we can combine all three of these strategies linearly just as in the proof of Theorem~\ref{CHSHAND}.  Note that, due to the cubic term $XYZ$ we will need to play a cubic number of GHZ games in secret to simulate rounds where the inputs are linear combinations of previous rounds.  As a result the final entropy bound will have a cubic blow up (rather than a quadratic blow up as in the proof for CHSH).  The final entropy bound will be $H_0(A, B, C) \leq O(2^{3m})$.